\documentclass[]{llncs}

\usepackage{algorithmic,algorithm}
\usepackage{graphicx}
\usepackage{amssymb}
\usepackage{amsmath}
\usepackage[gen]{eurosym}

\newtheorem{thm}{Theorem}
\newtheorem{defn}[thm]{Definition}
\newtheorem{lem}[thm]{Lemma}

\newcommand{\tuple}[1]{\langle #1 \rangle}

\begin{document}

\titlerunning{Compositional Verification of Evolving SPL}

\title{Compositional Verification of Evolving Software Product Lines}

\authorrunning{J-V Millo, S Ramesh, S Krishna, Ganesh K}

\author{Jean-Vivien Millo\inst{1}\inst{2} and S Ramesh\inst{2} and Shankara Narayanan Krishna\inst{3} and Ganesh Khandu Narwane\inst{4}}

\institute{EPI AOSTE, INRIA Sophia-Antipolis, France
\and Global General Motors R\&D, TCI Bangalore, India
\and Department of CSE, IIT Bombay, Mumbai, India
\and Homi Bhabha National Institute, Mumbai, India}

\maketitle

\begin{abstract}
This paper presents a novel approach to the design verification of Software Product Lines(SPL). 
The proposed approach assumes that the requirements and designs are modeled as finite state machines
with variability information. The variability information at the requirement and design levels are 
expressed differently and at different levels of abstraction. Also the proposed approach supports
verification of SPL in which new features and variability may be added incrementally. 
Given the design and requirements of an SPL, the proposed design verification method ensures that 
every product at the design level behaviorally conforms to a product at the requirement level. 
The conformance procedure is compositional in the sense that the verification of an entire SPL 
consisting of multiple features is reduced to the verification of the individual features. 
The method has been implemented and demonstrated in a prototype tool SPLEnD (SPL Engine for Design Verification) 
on a couple of fairly large case studies.
\end{abstract}





\section{Introduction}
\label{sec-intro}
Large industrial software systems are often developed as 
{\em Software Product Line} (SPL) with a common core set of features which are developed once and reused across 
all the products.
The products in an SPL differ on a small set of features which are specified using {\em variation points}.  
The focus of this paper is on modeling and analysis of SPLs which have drawn 
the attention of researchers recently~\cite{benavides10-is,Classen2011,Cordy2012}. 

Many approaches have been proposed to describe SPLs, the most prominent one 
being {\em feature diagrams}.
All these proposals seem to assume a global view of SPL as they start 
with a complete list of features and the variation points using a single 
vocabulary. 
All the subsequent SPL assets, like requirement documents, design models, 
source codes, test cases, documentations, share the same definition and 
vocabulary~\cite{Czarnecki00,1248997}. 
The assumption of a single homogeneous and global view of variability 
description is inapplicable in many practical settings, where 
there is no top level complete description of features and variabilities. They 
often evolve during the long lifetime of an SPL as new features and 
variabilities are added during the evolution. 
Further, SPL developers tend to use different representations and vocabulary of variability at different stages of 
development: at the requirement level, a more abstract and intuitive 
description of variation points are used, while at the 
design level, the efficiency of implementation of variation points is of primary 
concern.
For example, consider the case of an automotive SPL, where 
one variation point is the region of sale (eg. Asia Pacific, Europe, 
North America etc). At the requirement level, this 
variation point is expressed directly as an enumeration variable assuming one
value for every region. Whereas, at the design level, the 
variation point is expressed using two or three boolean variables; by setting 
the values of the boolean variable appropriately, the behavior 
specific to a region is selected at the time of deployment. 

We present a design verification approach that is more suited to the above kind of evolving 
SPLs in which different representation of variabilities would be used at the requirement 
and design level.  One natural and unique problem that arises in this context is to 
relate formally the variation points expressed at different levels of
abstractions.
Another challenge is the analysis complexity: the number of products is exponential in 
the number of variation points and hence product centric analyses are not 
scalable. We propose a compositional approach in which every 
feature of the SPL is first analyzed independently; the per-feature analysis 
results are then combined to get the analysis result for the whole SPL. 

For capturing variability in the behavior of an SPL, we have extended 
the standard finite state machine model, which we call 
{\em Finite State Machines with Variability}, in short, {\em FSMv}.
The behavior and variability of a feature at the requirement and design level 
can be modeled using FSMv. 
We define a conformance relation between FSMvs to relate the requirement and design models.
This relation is based upon the standard language containment of state machines.

One unique feature of FSMv is that it provides a compositional operator for composing 
the feature state machines to obtain a model for an SPL. This operator thus enables 
incremental addition of features and variabilities. The proposed verification approach 
exploits the compositional structure of the SPL models to contain the analysis complexity. 

\begin{figure}[b]
 \begin{center}
   {\includegraphics[width=0.6\textwidth]{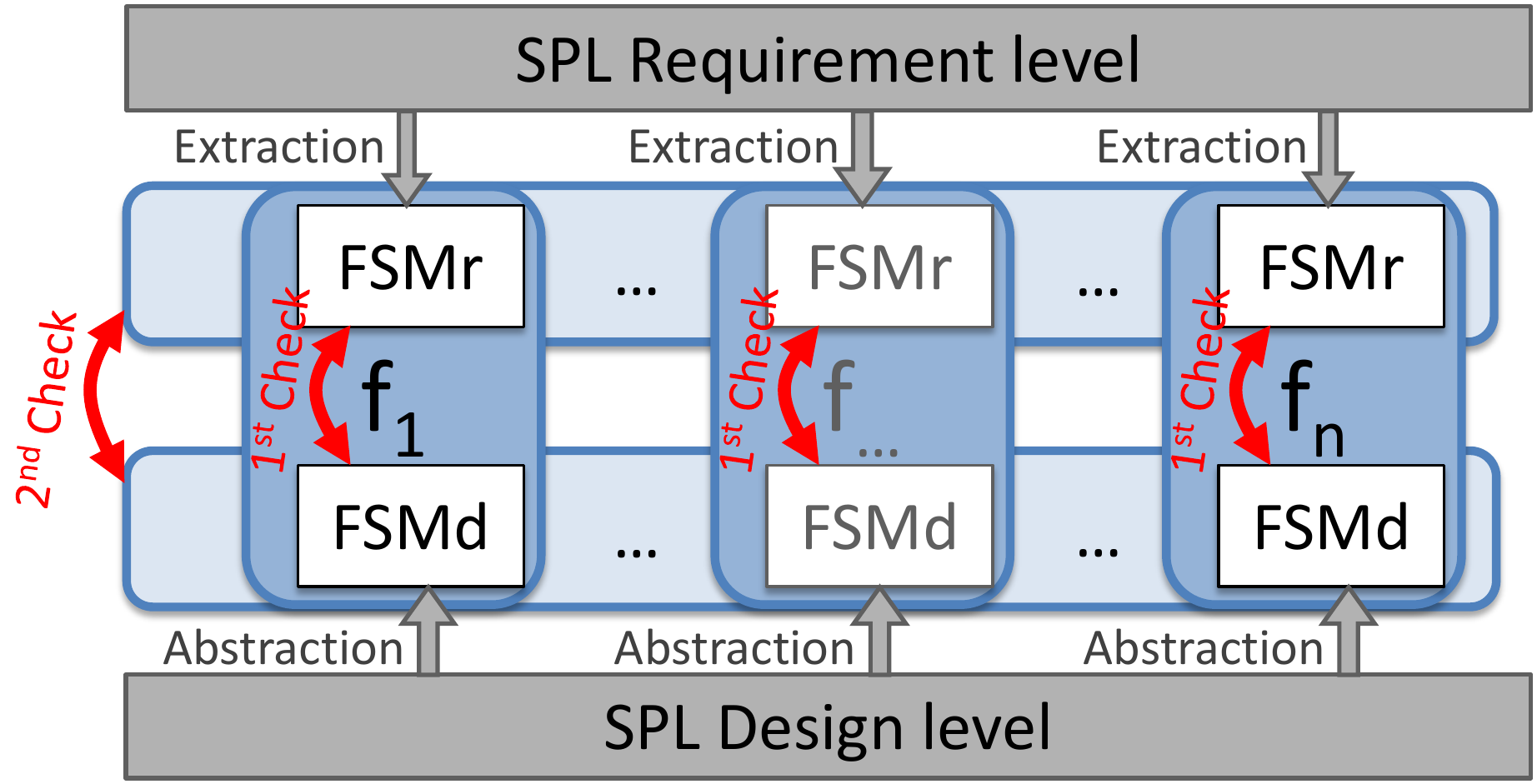}}
  \caption{The proposed verification framework.} 
  \label{figure_fsmvv}
  \end{center}
\end{figure}

Figure \ref{figure_fsmvv} summarizes the proposed approach. It shows an SPL composed of 
features $f_1$ to $f_n$. 
Each feature has an FSMv model of its requirements (called FSMr) and an FSMv model
derived from its design (called FSMd). The proposed analysis method checks whether the FSMd of every feature conforms to its 
FSMr ($1^{st}$ check). The output of this first step is a conformance relation between each 
pair of FSMr and FSMd. The obtained conformance relations are then used to check 
whether the actual behavior of the entire SPL conforms to the expected 
one ($2^{nd}$ check).  We reduce this check to checking the satisfiability of a 
Quantified Boolean Formula (QBF). There is no need to build the entire behavioral 
model of the SPL in the second step.  

We have built a prototype tool SPLEnD based upon this approach. This tool performs 
the first check using SPIN \cite{spinbook} while the the well-known QBF SAT solver
CirQit \cite{GoultiaevaB10} is used for the second step. We have experimented with 
the tool using modest industrial size examples with very encouraging results. 
An earlier version of this work (October 2012) can be found at \cite{HAL}.
\subsection{Related works}
FSMv and the proposed design verification approach were developed 
independently but has some apparent similarities with the 
FTS$^+$ model \cite{Classen2011}, which  also extends finite state machines to 
include certain product variability information. 
However, there is a motivational difference between the two formalisms. 
The aim of FTS$^+$ is to model the entire SPL 
and hence there is a single global machine with a single global vocabulary for 
expressing variabilities; the variability information represents the 
presence/absence of features in the SPL. 
In contrast, our approach is based upon a differnt 
view of SPL: a feature with variability is an increment in functionality 
and an SPL is a collection of features. We use a single
FSMv to model a feature and a whole SPL is modeled as a parallel 
composition of FSMv machines. 

The difference in viewpoint has another consequence: FTS$^+$ 
models, since they model the entire SPL, tend to be  large and hence 
has high analysis complexity. Efficient abstraction techniques 
are hence used for solving this problem~\cite{Cordy2012}. Whereas, each 
FSMv models a fraction of functionality and hence can be 
analysed easily. Further, the entire SPL can be modeled as 
composition of FSMvs and can be efficiently analysed using composition 
techniques.

Many other behavioral models have also been 
proposed \cite{LarsenNW07,Benveniste09,FantechiG08,Scheidemann08b}
which are usually coupled with a variability model such as OVM \cite{1248997}, the 
Czarnecki feature model \cite{Czarnecki00}, or VPM \cite{gomaa2008} to attain a fair level 
of variability expressibility. Unlike all these approaches, 
FTS$^+$ \cite{Classen2011} and FSMv capture the variability in an explicit way which we 
find more intuitive.

The Variation Point Model (VPM) of Hassan Gomaa \cite{gomaa2008} distinguishes 
between variability at the requirement and design levels but no design verification 
approach has been presented. Kathrin Berg {\em et al.}\cite{Berg2005} propose a model for 
variability handling throughout the life cycle of the SPL. Andreas 
Metzeger {\em et al.}\cite{Metzger2007} and M Riebisch {\em et al.}\cite{Riebisch2008} 
provide a similar approach but they do not consider the behavioral aspect. In the proposed approach, 
we extract the relation between requirement and design level variability from a behavioral analysis.

Kathi Fisler {\em et al.} \cite{Fisler2007} have developed an analysis based on 
three-valued model checking of automata defined using step-wise refinement. Later on, 
Jing Liu {\em et al.} \cite{lutz2011} have revisited Fisler's approach to provide a much more 
efficient method. Recently, Maxime Cordy {\em et al.} have extended Fisler's approach to LTL formula
\cite{CordySPLC2012}.
Kim Lauenroth {\em et al.} \cite{lauenroth2011a} as well as 
Andreas Classen {\em et al.} \cite{Classen2011,Cordy2012}, and 
Gruler {\em et al.} \cite{Scheidemann08a} have developed model checking methods for SPL 
behavior. These methods are based on the verification of LTL/CTL/modal $\mu$ calculus 
formula.  

All these verification methods assume a global view of variability and hence 
the representation of variability information is identical in both 
specification and the design.  In contrast, in our work the specification and 
design involve variability information at different levels of abstraction and 
hence one needs mapping information between the two levels. Furthermore, 
our formalism allows incremental addition of functionality and variability and enables
compositional verification.

\section{Design Verification of a Single Feature}
\label{sec_fbv}
An SPL, in general consists of multiple features, each feature having 
different functionality and variability.  A typical body control software of
an automotive system is an SPL that has several features such as door lock, 
lighting, seat control etc. Each of these features has a distinct function and 
variability. For example, the locking behaviour of a door lock function has a 
variation point called {\em transmission type}. If the transmission type is 
manual then the door is locked after the speed of the vehicle exceeds a certain 
threshold value; for automatic transmission, the door is locked when the gear 
position is shifted out of park. 
In this section we will focus on modeling and relating the design of a {\it single feature}
to its requirement.  
\subsection{FSMv and language refinement}
\label{ssec_ftf}
{\it Finite State Machines with Variability (FSMv)} is an extension 
of finite state machines, to represent all possible behaviours of a 
feature. 
Let $Var$ be a finite set of variables, each taking a value ranging 
over a finite set of values. 
Let $x \in Var$, and let $Dom(x)$ be the 
finite set of values that $x$ can take. 
The set of atomic formulae 
we consider are $x = a$, $x \neq a$,  
$x = y$, $x \neq y$ for $a \in Dom(x)$, and $x, y \in Var$. 
Let $A_{Var}$ denote the set of atomic formulae over $Var$. 
Let $\alpha$ represent a typical element of $A_{Var}$.
Define 
$$\Delta::=\alpha~|~\neg \Delta~|~\Delta \wedge \Delta~|~\Delta \vee \Delta~|\Delta \Rightarrow \Delta$$  to be the set 
of all well formed predicates over $Var$. 

\begin{defn}[FSMv]
An FSMv is a tuple ${\cal A}=\tuple{Q,q_0,\Sigma,Var,E,\rho}$ where:\\
(1) $Q$ is a finite set of states; $q_0$ is the initial state; (2) $\Sigma$ is a finite set of events;	
(3) $Var$ is a finite set of variables;
	(4) $E \subseteq Q\times \Delta \times \Sigma \times Q$ gives the set of transitions. 
	A transition $t=(s, g, a, s')$ represents a transition from state $s$ to state $s'$ 
	on event $a$; the predicate $g$ is called a guard of the transition $t$; 
	$g$ is consistent and defines the variability domain of the transition; 
	(5) $\rho \in \Delta$ is a consistent predicate called the global predicate. 
\end{defn}
The variables in $Var$ determines the variability allowed in the feature with 
each possible valuation of the variables corresponding to a variant. The 
allowed values of the variables are constrainted by the global predicate 
$\rho$.  For example, if $\rho$ is $((x=1) \vee (x=2)) \wedge (x=y-1)$,
then the allowed variants are those for which the values for the pairs $(x,y)$ 
are $(1,2), (2,3)$. 
The predicate in a transition determines the variants to which the transition 
is applicable. 
While drawing a transition 
 $t=(s, g, a, s')$, the edge connecting $s$ to $s'$ is decorated with  $g:a$.
When $g$ is true, we simply write $a$ on the edge.
\begin{defn}[Configuration]
\label{def_conf}
A configuration, denoted by $\pi$, is an assignment of values 
to the variables in $Var$. The set of all configurations is denoted by $\Pi_{Var}$, or $\Pi$, when $Var$ is clear from the context.
Define $\Pi(\rho)=\{\pi \mid \pi \models \rho\}$ to be the set 
of all those configurations that satisfy $\rho$. 
The elements of  $\Pi(\rho)$ are called valid configurations. 
Given a valid configuration $\pi$ and a transition $t=(s,g,a,s')$, we say that $t$ 
is enabled by $\pi$ if $\pi \models g$. 
\end{defn}

As a concrete example of an FSMv, consider the feature 
{\em Door lock} in automotive SPL which controls the 
locking of the doors when the vehicle starts.
The expected behavior of this feature is modeled using 
the FSMv $Req_{dl}$ described pictorially in Figure \ref{figure_adlreq}. 
In the initial state, this feature becomes active 
when all the doors are closed. The doors are locked when either the speed of 
the vehicle exceeds a predefined value or the gear is shifted out of park.
An unlock event reactivates the feature. 
There are
four configurations for this feature all of which are described using the three 
variables: $DL\_Enable$, $Transmission_{dl}$ and $DL\_User\_Pref$. 
The top box denotes the values that these variables can assume, and 
the bottom box gives the global predicate ($\rho$) associated with the machine. 
$\rho$ ensures that in every 
valid configuration, the variable $Transmission_{dl}$ having the value $Manual$ implies 
that $DL\_User\_Pref$ takes the value $Speed$. This captures the fact that in 
manual transmission, there is no park position on the gearbox. 
To avoid clutter, we have replaced guards of the form $x=i$ with $i$ in the figure. 
The transition labeled with $Disable:*$ means that when 
$DL\_Enable$ assumes the value $Disable$, it stalls on any event. 

\begin{figure}[hbpt]
 \begin{center}
   {\includegraphics[width=0.6\textwidth]{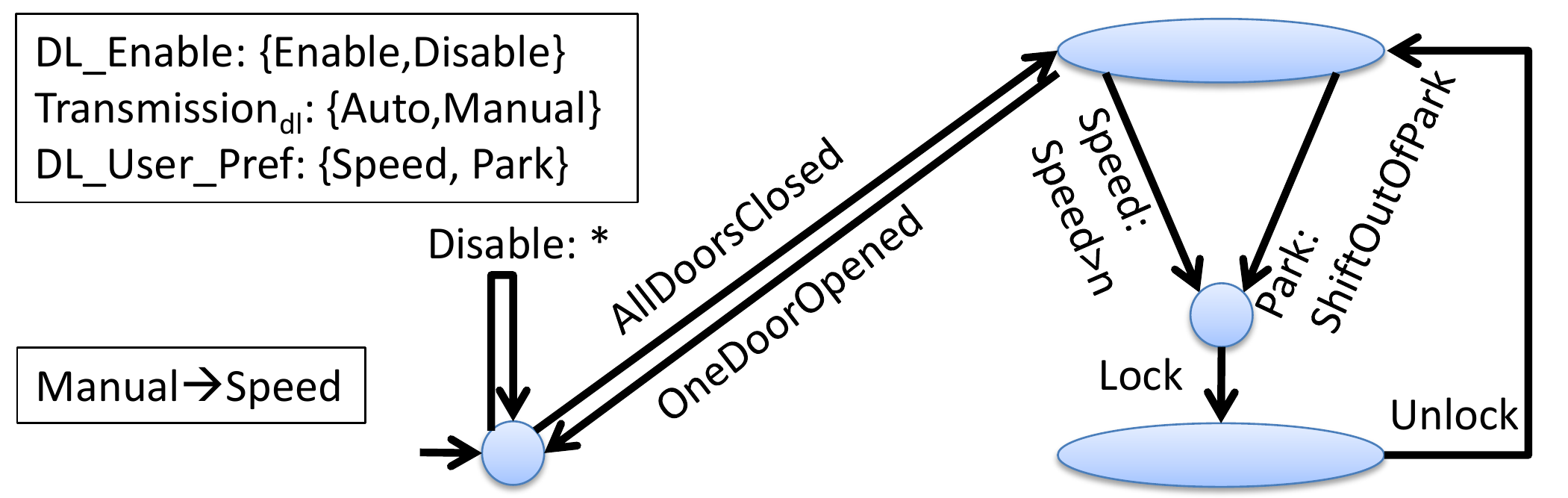}}
  \caption{The FSMv of the feature {\em Door lock}.} 
  \label{figure_adlreq}
  \end{center}
\end{figure}

\subsubsection{Requirement against Design}
\label{sssec_fsmdr}
In the requirement of a product line, the variability is usually discussed in 
terms of variation points, which are at a high level of abstraction and 
focused on clarity and expressibility. The restriction of the possible 
configurations is expressed as general constraints on these variation 
points, e.g., the global predicate $Manual\implies Speed$ in 
the {\em Door lock} example.  In contrast, in a design, the variability 
description is constrained by efficiency, implementability, ease of 
reconfiguration and deployment considerations. For instance, in the automotive 
applications, one often finds {\em calibration parameters} ranging over a set of boolean values.
Further, the constraint on the calibration parameters ($\rho$) takes the 
special form of the list of the possible configurations of the calibration 
parameters in order to easily configure the design. 

FSMv can capture both the design as well as the requirements of a feature.
 We distinguish the requirement and design models 
by denoting them FSMr and FSMd respectively.  
Figure \ref{figure_adlreq} presents the FSMr, $Req_{dl}$, of the feature {\em Door lock}.
The  FSMd, $Des_{dl}$, of the feature {\em Door lock} is presented in 
Figure \ref{figure_adldesign}. 
The structure of $Des_{dl}$ is similar to $Req_{dl}$ except 
that the top elliptical shaped state in Figure \ref{figure_adlreq} is split 
into two states (the top and the bottom elliptical shaped states) in Figure~\ref{figure_adldesign}. 
The top state is for auto-transmission
whereas the bottom one is for manual transmission as can be seen from the 
configuration label of the two transitions going from the initial state. 
Two variables $Cp1$ and $Cp2$ encode the possible 
configurations in the FSMd.
The box in  Figure \ref{figure_adldesign} 
depicts the set of possible values of these. 	   
$Cp1=Auto$ corresponds to the configuration in 
which the transmission is $Auto$ whereas $Cp1=Moff$ corresponds to either the 
manual transmission or the case when $Cp1$ 
is disabled;  
similarly, $Cp2=Speed$ means that the user preference is set on $Speed$,  
while $Cp2=Poff$ means either $Park$ or the case when $Cp2$ is disabled. 
\begin{figure}[hbpt]
 \begin{center}
   {\includegraphics[width=0.4\textwidth]{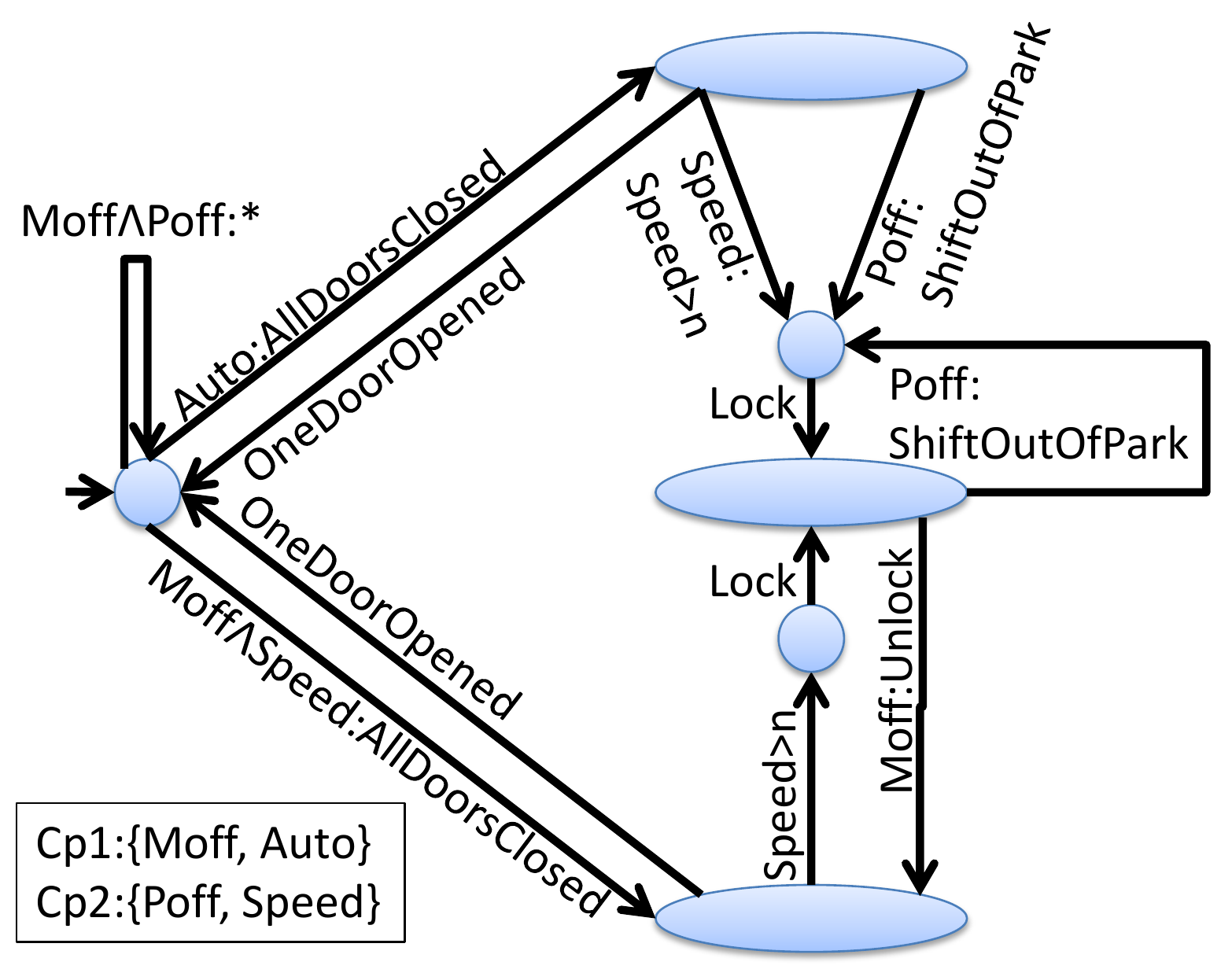}}
  \caption{$Des_{dl}$: the FSMd abstracted from the design of the feature {\em Door lock}.} 
  \label{figure_adldesign}
  \end{center}
\end{figure}


\subsection{Variants of FSMv and Conformance}
Having described the design and requirement behaviour of a feature $f$ 
using FSMd and FSMr respectively, we now define the notions of variants 
and conformance. A variant of an FSMv corresponds to one of the several 
possible behaviours of the feature (at the design, requirement level respectively).
Given a feature $f$, and a (FSMd, FSMr) pair corresponding to $f$, 
we say that the design of $f$ conforms to the requirements of $f$ 
provided every variant of the FSMd has a corresponding 
FSMr variant.


\begin{defn}[Variant of an FSMv]
\label{defn_conf22}
Let ${\cal A}=\tuple{Q,q_0,\Sigma,Var,E,\rho}$ be an FSMv and 
$\pi\in \Pi(\rho)$ be a valid configuration of ${\cal A}$. A variant of 
${\cal A}$ is an FSM obtained by retaining only  
transitions $t=(s,  g, a, s')$, and states $s,s'$ such that $g \models \pi$. 
Once the relevant states and transitions are identified, 
we remove the guards $g$ from all the transitions; 
$\rho$ is also removed.
The resultant FSM is denoted  ${\cal A} \downarrow \pi$.
\end{defn}

In the example of FSMr for the feature {\em Door lock}, the variant
$Req_{dl} \downarrow \tuple{Enable,Auto,$ $Park}$ does not contain 
the transitions with the event $Speed>n$ and $*$. 
We compare  the FSMd and FSMr of a feature $f$ using their variants. 
Given an FSMv  ${\cal A}$, we associate with each configuration $\pi$ of ${\cal A}$ 
 the language of the FSM ${\cal A}\downarrow \pi$, 
denoted by $L({\cal A}\downarrow \pi)$. 
\label{sssec_ref}
We say that an FSMd ${\cal A}_d$ conforms to an FSMr ${\cal A}_r$ if and only if 
the behaviour of every variant of ${\cal A}_d$ is contained 
in the behaviour of some variant of ${\cal A}_r$. 
 
\begin{defn}[The conformance mapping $\Phi$]
\label{def_mapcalvp}
Let ${\cal A}_r$ and ${\cal A}_d$ be a pair of FSMr and FSMd respectively with global 
predicates $\rho^d$ and $\rho^r$. Let $\Pi_d, \Pi_r$ be the set of all 
design, requirement configurations. Then ${\cal A}_d$ conforms to ${\cal A}_r$ 
denoted ${\cal A}_d \leq_{\Phi} {\cal A}_r$ 
if  there exists a mapping $\Phi: \Pi_d(\rho^d) \to 2^{\Pi_r(\rho^r)}$ 
such that $\forall \pi_d \in \Pi_d(\rho^d), \exists \pi_r \in \Pi_r(\rho^r)$ satisfying 
$L({\cal A}_d \downarrow \pi_d) \subseteq  L({\cal A}_r \downarrow \pi_r)$. 
$\Phi$ is called the conformance mapping.
\end{defn}


In the feature {\em Door lock}, $\Phi(\tuple{Moff,$ $Speed})$ contains $\tuple{Enable,Manual,Speed}$
 since $L(Des_{dl}\downarrow \tuple{Moff,$ $Speed}) \subseteq L(Req_{dl}\downarrow \tuple{Enable,Manual,Speed})$.

\subsection{Checking the conformance}
\label{ssec_fmusingfsmv}
Let $f$ be a feature with FSMr $Req_{f}$ and FSMd $Des_{f}$.
Then the conformance checking problem is to compute a mapping $\Phi$ such that 
$Des_{f}\le_\Phi Req_{f}$.

The conformance mapping is computed by comparing every projection of 
$Des_f$ with every projection of $Req_f$. 
Algorithm 1, given below, presents a possible implementation using the standard
automata containment algorithm\cite{vardi-wolper}, as implemented in the SPIN 
model checker \cite{spinbook}. To use SPIN, one should describe the system
along with the checked property in the Promela language \cite{spinbook}.
Out of this description, SPIN generates the {\em pan.c} file 
which is the verifier for the system. After compilation, the {\em pan(.exe)} executable
performs the verification. \\

Algorithm 1 starts by generating a Promela file containing the definition of
(i) the environment, (ii) $Des_{f}$, (iii) $Req_{f}$, (iv) the initialization sequence and 
(v) a never claim which holds for the language containment condition. 
During the initialization, the configuration of 
$Des_{f}$ and $Req_{f}$ are initialized with a random couple of configurations. 
Then the environment, followed by $Des_{f}$ and $Req_{f}$ are run atomically. The never claim 
assertion is :  
$never(\box \diamond (\neg error(Des_{f}) \wedge error(Req_{f}))$, 
where $error(X)$ means that $X$ is in error state. The never claim is violated when 
the design is not in the error state but the requirement process is in 
the error state. This corresponds to a design configuration $\pi_d$
such that $Des_f \downarrow \pi_d$  handles an event, 
while $Req_f \downarrow \pi_r$  does not, for all possible requirement configurations $\pi_r$.
Algorithm~1 runs the full verification algorithm of SPIN for every pair $(\pi_d, \pi_r)$
 of design and requirement configurations.
SPIN(i.e. {\em pan(.exe)}) returns the list of pairs for which the conformance condition is violated.
Every other pair is added to the conformance mapping $\Phi$.
Lemma \ref{lemma:spin} proves the correctness of Algorithm 1.

\begin{algorithm}[htb]
\label{algo_check_feature_language}
\caption{implements the conformance checking using SPIN.}
\begin{algorithmic}
\STATE \textbf{Input :} $Des_f$, $Req_f$.\\
\STATE \textbf{Output :} The mapping $\Phi$ when $Des_{f}\le_{\Phi}Req_{f}$\\ 
\STATE 1. Generate a Promela file which contains $Req_{f}$, $Des_{f}$, the environment, the never claim, and the initialization sequence.
\STATE 2. Launch the full verification algorithm of spin
\STATE 3. Build the mapping $\Phi$ from the output of spin.
\STATE 4. Conclude whether the design conforms to the requirement
\IF{$\forall \pi_d\in \Pi(\rho_d)$, $\Phi(\pi_d)\ne \emptyset$}
\RETURN $true$ along with ($\Phi$) 
\ELSE 
\RETURN $false$ along with ($\pi_d$) \COMMENT{where $\pi_d$ has no correspondence through $\Phi$}
\ENDIF
\end{algorithmic}
\end{algorithm}

\begin{lem}
\label{lemma:spin}
Given FSMd $Des_f$ and FSMr $Req_f$ for a feature $f$, let $(\pi_d,\pi_r)$ be a pair of design and requirement configurations. 
Then, $L(Des_{f}\downarrow \pi_d) \not\subseteq L(Req_{f}\downarrow \pi_r)$ if and only if 
$\neg error(Des_{f}) \wedge error(Req_{f})$.
\end{lem}
\begin{proof}
Assume $L(Des_{f}\downarrow \pi_d) \not\subseteq L(Req_{f}\downarrow \pi_r)$. 
Then there exists a word $w \in L(Des_f\downarrow \pi_d)$ which is prefixed by $u.e$,  
with $u$ a finite prefix of a word in $L(Req_f \downarrow \pi_r)$, and 
$e$ an event such that $u.e$ is not a prefix of any word in 
$L(Req_f\downarrow \pi_r)$. In such a situation, $Des_{f}$ does not go 
to the error state but $Req_{f}$ does.\\

Conversely, if $L(Des_f\downarrow \pi_d) \subseteq L(Req_f\downarrow \pi_r)$, then
whenever $Des_f$ is not in an error state, 
 $Req_f$ will also not be in an error state.
\qed 
\end{proof}


\section{Design Verification of SPL}
\label{sec_splbv}
In the previous section, we looked at individual features in an SPL and provided a method for
 comparing the design and requirements of a feature, both containing variabilities. 
In this section, we extend this method to verifying a whole SPL design against its requirements. 
An SPL is essentially a composition of multiple features 
satisfying certain constraints.
We define a parallel composition operator over FSMv to model an SPL. 
The features in an SPL can interact and we  follow one of the
standard methods of allowing the composed FSMv models to share some 
common events, which  correspond to
two-party handshake communication events. 
A distinguishing aspect of the proposed parallel operator is that it takes into account the constraints across
 the composed machines. The constraints
could be of various types, e.g.  dependency and exclusion relations, and are 
modeled as predicates over variables of the composed features.
 \begin{defn}[Parallel composition of FSMv]\-\\
\label{def_compo_fsmc}
Let ${\cal A}_x=\tuple{Q_x,q_0^x,\Sigma_x,Var_x, E_x, \rho_x}$,  $x\in\{1,2\}$ be two FSMv's with 
$Var_1 \cap Var_2=\emptyset$. 
Let $H = \Sigma_1 \cap \Sigma_2$ be the set of handshaking events.
Let $\rho_{12}$ be a predicate over $Var_1 \cup Var_2$, 
such that $\rho_{12} \wedge \rho_1 \wedge \rho_2$ is consistent. 
$\rho_{12}$ is the composition predicate capturing the possible constraints between 
the variabilities of the two composed features. 
Let $\rho=\rho_{12} \wedge \rho_1 \wedge \rho_2$. 

The parallel composition of ${\cal A}_1$ and ${\cal A}_2$ denoted by 
${\cal A}={\cal A}_1 \parallel {\cal A}_2$ is a tuple  
$\tuple{Q_1\times Q_2,(q_0^1,q_0^2),\Sigma_1\cup \Sigma_2, Var_1\cup Var_2, E, 
\rho}$ with transitions defined as follows: Consider a state  $(s_1, s_2) \in  Q_1\times Q_2$, 
and transitions  $(s_1,g_1,a_1,s'_1) \in E_1$ and $(s_2,g_2,a_2,s'_2) \in E_2$.\\
(1) If $a_1=a_2=a \in H$, define $((s_1,s_2), g_1 \wedge g_2, a, (s'_1, s'_2)) \in E$, 
provided $g_1 \wedge g_2$ is consistent and $g_1 \wedge g_2 \models \rho$. \\
(2) If $a_1 \in \Sigma_1 \backslash H$, 
define $((s_1,s_2), g_1, a_1, (s'_1, s_2)) \in E$, $g_1 \models \rho$.\\
(3) If $a_2 \in \Sigma_2 \backslash H$, 
define $((s_1,s_2), g_2, a_2, (s_1, s'_2)) \in E$, $g_2 \models \rho$.
\end{defn}


For illustration, consider the feature {\em Door unlock} which 
automates the unlocking of the doors in a vehicle.
Figure \ref{figure_adu}-a gives the FSMr of the feature extracted from the requirements.
From the initial state, the feature becomes 
active when the event $Lock$ happens. As soon as either the key 
is removed from ignition or the gear is shifted to park position, the doors get unlocked 
and the feature {\em Door unlock} becomes inactive.
Figure \ref{figure_adu}-b presents the FSMd of the feature {\em Door unlock}.
It is quite similar to the requirement except that the active state is split in two:  
the feature reacts to the {\em ignition Off} event in one state, and 
to the {\em Shift Into Park} event in another state.

Let us consider the composition of the two FSMrs of the features {\em Door lock} and {\em Door unlock}.  
The handshake events between the two features are {\em Lock} and {\em Unlock}. 
In the composition, we introduce the following composition  
predicate: $(DU\_Enable = Enable \Leftrightarrow DL\_Enable = Enable) 
\wedge Transmission_{dl} = Transmission_{du}$, 
which brings out the natural constraints that {\em Door lock} feature is enabled
if and only if {\em Door unlock} is also enabled and the 
transmission status has to be the same.

\begin{figure}[hbpt]
 \begin{center}
   {\includegraphics[width=0.3\textwidth]{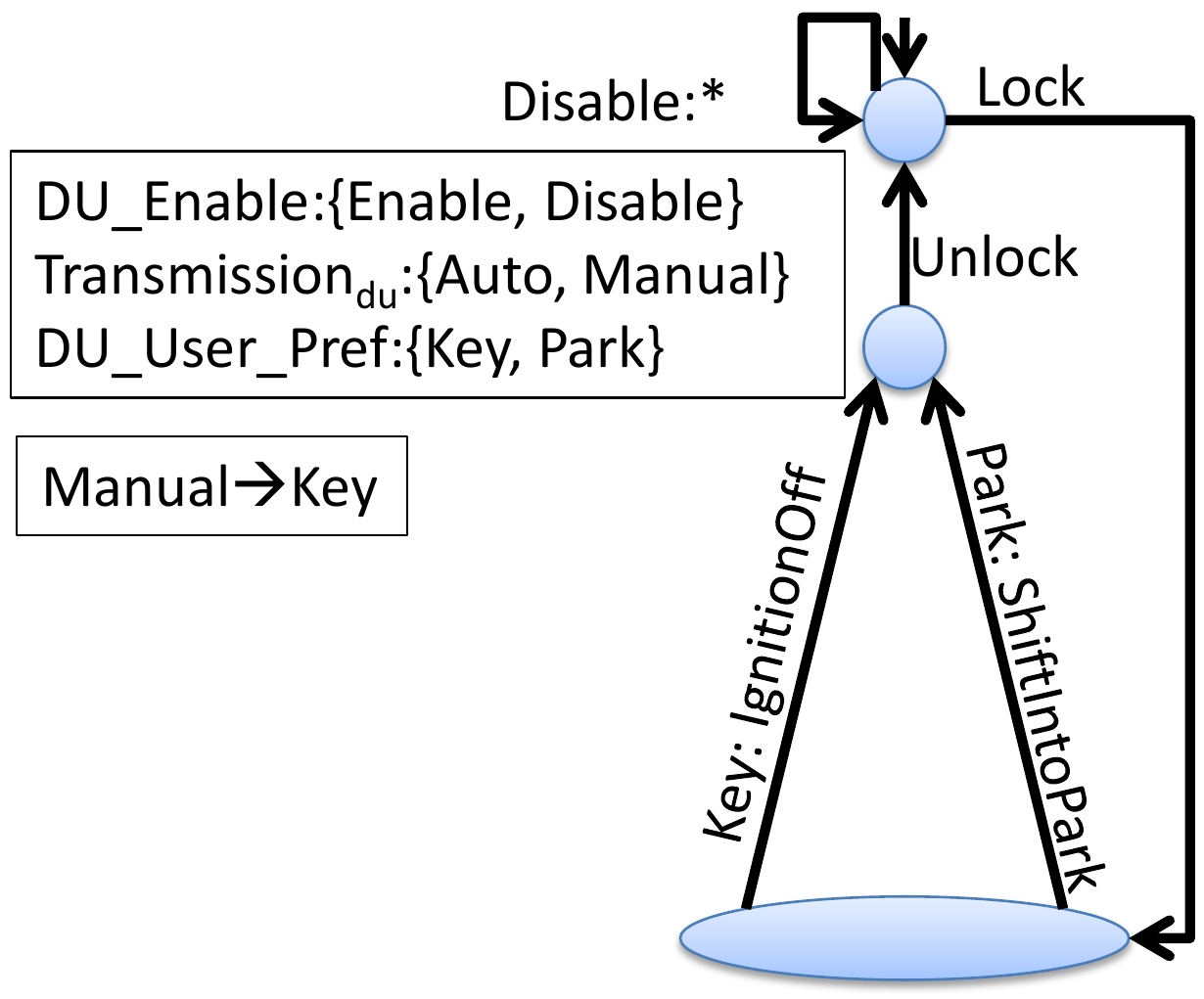}}~~~~~~
   {\includegraphics[width=0.3\textwidth]{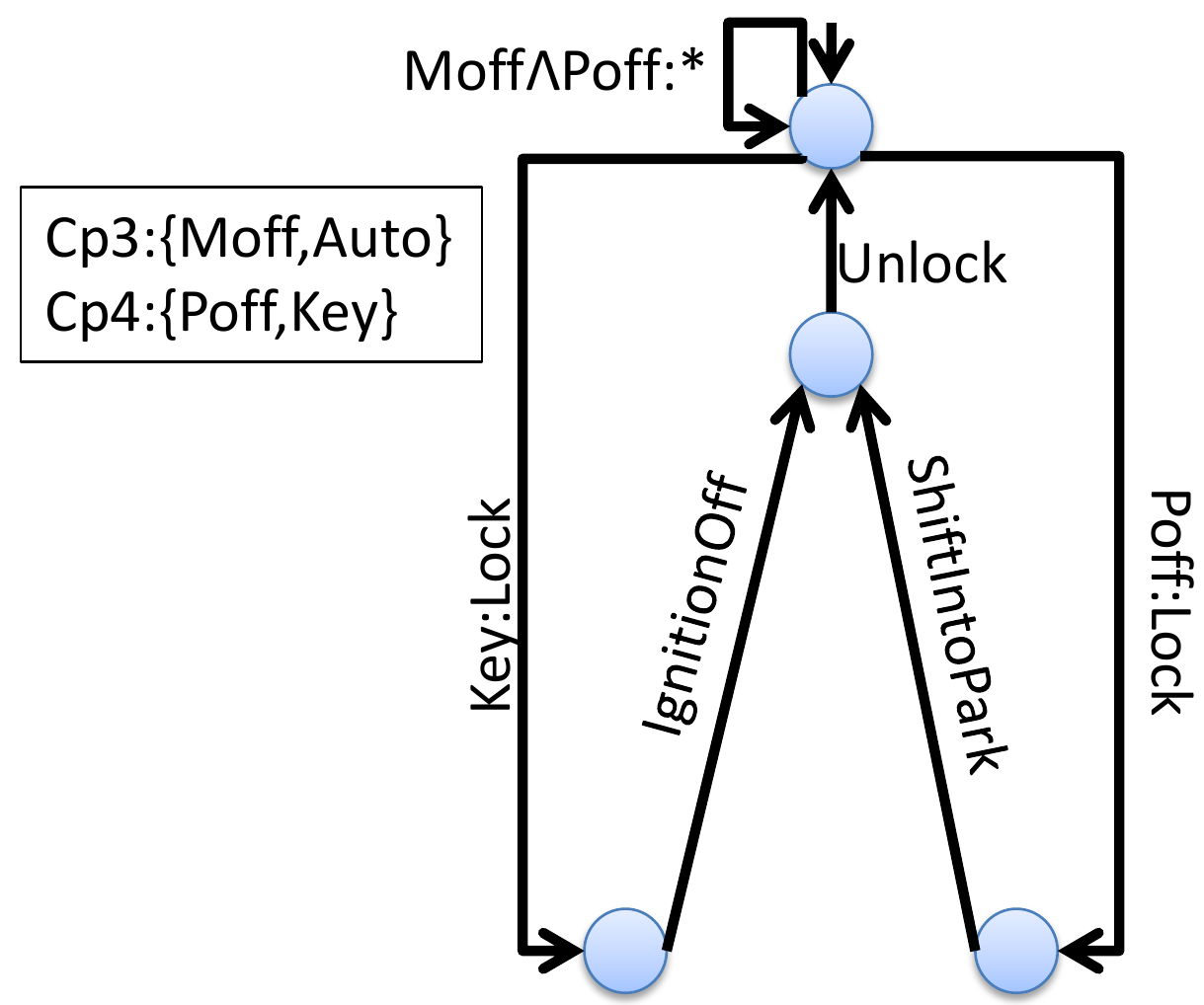}}\\
   a)~~~~~~~~~~~~~~~~~~~~~~~~~~~~~~~~~~~~~~~~~
   b)
  \caption{a) $Req_{du}$: the {\em Door unlock} FSMr  and b) $Des_{du}$: the corresponding FSMd.} 
  \label{figure_adu}
  \end{center}
\end{figure}
The valid configurations after composition are restricted by the composition predicate.
We provide a few definitions to define composite valid configurations. 

\begin{defn}[Composing Configurations]
\label{not_conf_comp}
Let ${\cal A}_i=(Q_i, q_0^i, \Sigma_i, Var_i, E_i, \rho_i)$ be two FSMv's, and 
let ${\cal A}={\cal A}_1 \parallel {\cal A}_2$ be as given by definition \ref{def_compo_fsmc}. 
Let $\rho=\rho_{12} \wedge \rho_1 \wedge \rho_2$ be the global predicate 
of ${\cal A}$. Consider  two valid configurations  $\pi_1\in \Pi(\rho_1)$ and $\pi_2\in \Pi(\rho_2)$ 
of ${\cal A}_1$ and ${\cal A}_2$. 
The compostion of $\pi_1, \pi_2$, denoted $\pi_{12}$ 
is a configuration over $Var_1 \cup Var_2$ such that 
$\pi_{12}$ agrees with $\pi_1$ over $Var_1$, and 
agrees with $\pi_2$ over $Var_2$, and $\pi_{12} \models \rho$. $\pi_{12}$ is a valid configuration 
of ${\cal A}$ and we denote it by $\pi_{12}=\pi_1+\pi_2$.
\end{defn}

\begin{lem}
\label{lemma:comp}
Let ${\cal A}_1$ and ${\cal A}_2$ be two FSMv's.
For each valid configuration $\pi$ of ${\cal A}_1 \parallel {\cal A}_2$, 
there are valid configurations $\pi_1$ of ${\cal A}_1$ and  $\pi_2$ of ${\cal A}_2$ such that $\pi=\pi_1+\pi_2$. 
\end{lem}
\begin{proof}
 Let $\pi \in \Pi(\rho)$ with $\rho=\rho_{12} \wedge\rho_1\wedge\rho_2$ 
be a valid configuration of ${\cal A}_1 \parallel {\cal A}_2$.  
 $\rho_1$ and $\rho_2$ are the global predicates of 
${\cal A}_1$, ${\cal A}_2$ respectively, and $\rho_{12}$ is the composition 
predicate of ${\cal A}_1$, ${\cal A}_2$. By definition of valid configuration, 
$\pi \models \rho$; hence $\pi \models \rho_1$ and $\pi \models \rho_2$. 
Since $\pi$ is a configuration over $Var_1 \cup Var_2$, let us consider 
the restriction of $\pi$ on $Var_1$, call the resulting configuration $\pi_1$. 
Then $\pi_1 \models \rho_1$. Similarly, call  the restriction of $\pi$ 
on $Var_2$ as $\pi_2$. Then $\pi_2 \models \rho_2$. Then, $\pi_1, \pi_2$ 
are respectively valid configurations of ${\cal A}_1$ and ${\cal A}_2$. 
Hence, by definition \ref{not_conf_comp}, we obtain $\pi=\pi_1+\pi_2$. 
\qed
\end{proof}

In the example of feature $Door~Lock$,  the configuration $\tuple{Enable,Auto,Speed}$ from $Req_{dl}$ 
can be composed with $\tuple{Enable,$ $Auto,Key}$ from $Req_{du}$ because the transmission is $Auto$ in both 
(which is specified in the composition predicate).
$\tuple{Enable,Auto,Speed,Enable,Auto,Key}$
is a configuration of the parallel composition of $Req_{dl}$ with $Req_{du}$.

The parallel composition of FSMv's is such that each variant of the composition 
of two FSMv's is equal to the composition of 
 variants of the individual FSMv's.  

\begin{lem}[Variants of a composed FSMv]
\label{lemma:comp1}
Let ${\cal A}_1$ and ${\cal A}_2$ be two FSMv machines. Let $\pi$ be a valid configuration of 
${\cal A}_1\parallel {\cal A}_2$.
Then $L([{\cal A}_1\parallel {\cal A}_2] \downarrow \pi)$ = 
$L({\cal A}_1\downarrow \pi) \parallel L({\cal A}_2\downarrow \pi)$. 
 \footnote{The right hand side $\parallel$ refers to the standard communicating finite state machine composition.}
\end{lem}
 \begin{proof}
\label{proof:comp1}
We review some preliminary definitions before the proof. In the following, the operation $\parallel$ 
stands for (i) shuffle of words, (ii) shuffle of languages, (iii)parallel composition of FSMs, and 
(iv) parallel composition of FSMv. The context is clear in each case; hence 
there is no confusion.

\begin{defn}
\label{basic1}
 Let $\Sigma_1, \dots, \Sigma_n$ be $n$ finite sets of symbols. Let $\Sigma$ be a finite set. 
Given a word $w \in \Sigma^*$, we denote by $w \downarrow \Sigma_i$, 
the unique subword of $w$ over $\Sigma_i^*$. For example, 
if $\Sigma_1=\{a,b,e\}, \Sigma_2=\{a,e,f\}$, and if we consider $w=aefedefr \in \{a,d,e,f,r\}^*$, 
then $w \downarrow \Sigma_1=aeee$ and $w \downarrow \Sigma_2=aefeef$.  
\end{defn}

\begin{defn}(Asynchronous Shuffle)
\label{prod}
Let $\Sigma_1, \dots, \Sigma_n$ be $n$ finite sets. Let $\Sigma=\cup_{i=1}^n \Sigma_i$. 
Consider $n$ words $u_1, u_2, \dots, u_n$, $u_i \in \Sigma_i^*$. 
The asynchronous shuffle of $u_1, \dots, u_n$ denoted $u_1 \parallel \dots \parallel u_n$ 
is defined as $\{w \mid w \downarrow \Sigma_i=u_i\}$.
\end{defn}
As an example, consider $\Sigma_1=\{a,b,c,f\},\Sigma_2=\{a,d,e,f\}, \Sigma_3=\{c,d,f\}$, and 
 the words $u_1=abcf, u_2=adfe, u_3=dcf$. Then the word $w=abdcfe$ 
is in $u_1 \parallel u_2 \parallel u_3$ since, 
$w \downarrow \Sigma_i=u_i$ for $i=1,2,3$. Similarly, 
 the word $w'=adbcfe$ is also in  $u_1 \parallel u_2 \parallel u_3$.
 However, the word $w''=aebcfd$ is not in $u_1 \parallel u_2 \parallel u_3$, 
since $w'' \downarrow \Sigma_2=aefd$, not $u_2$.

The definition of shuffle can be extended from words to languages. 
We use the same notation $\parallel$ for the shuffle of sets, as 
well as for the shuffle of words. 

The asynchronous shuffle of two languages $L_1, L_2$ is defined as 
$L_1 \parallel L_2=\{w_1 \parallel w_2 \mid w_1 \in L_1, w_2 \in L_2\}$. For example,
if $L_1=\{abcf, abbf\}$ is a language over $\Sigma_1=\{a,b,c,f\}$ and  $L_2=\{adfe\}$ is a language 
over $\{a,d,e,f\}$, then $L_1 \parallel L_2=\{abcf \parallel adfe, abbf \parallel adfe\}$
=$\{abcdfe, adbcfe, abdcfe, abbdfe, abdbfe, adbbfe\}$.

\begin{defn}
\label{prod2}
Let $M_i=(Q_i, q_i, \Sigma_i, \delta_i)$ and 
$M_j=(Q_j, q_j, \Sigma_j, \delta_j)$ be complete FSMs.
The asynchronous product of $M_i, M_j$ is defined as the FSM 
$M_i \parallel M_j=(Q_i \times Q_j, (q_i, q_j), \Sigma_i \cup \Sigma_j, \delta)$ where 
\begin{enumerate}
 \item $\delta((q,q'),a)=(\delta_i(q,a), \delta_j(q',a)), a \in \Sigma_i \cap \Sigma_j$, 
\item  $\delta((q,q'),a)=(\delta_i(q,a), q'), a \in \Sigma_i, a \notin \Sigma_j$, 
\item $\delta((q,q'),a)=(q, \delta_j(q',a)), a \in \Sigma_j, a \notin \Sigma_i$. 
\end{enumerate}
On the common events, both FSMs move in parallel; otherwise, they move independent of each other.
\end{defn}
It is known that $L(M_i \parallel M_j)=L(M_i) \parallel L(M_j)$.  
Now we start the proof of Lemma \ref{lemma:comp1}.\\
Consider a valid configuration $\pi$ of ${\cal A}_1 \parallel {\cal A}_2$. 
As seen in Lemma \ref{lemma:comp}, we can find valid configurations $\pi_1$ 
of ${\cal A}_1$ and  $\pi_2$ 
of ${\cal A}_2$  such that $\pi=\pi_1+\pi_2$. 
The initial state of ${\cal A}_1 \parallel {\cal A}_2$ is $(q_0^1, q_0^2)$, where
 $q^0_1$ is the initial state of 
${\cal A}_1 $ and $q^0_2$ is the initial state of 
${\cal A}_2$. By definitions \ref{def_compo_fsmc}
and \ref{prod2},  if we consider a string 
$w=a_1 a_2 \dots a_n \in L[{\cal A}_1 \parallel {\cal A}_2] \downarrow \pi$, then we can find  strings
 $w_1  \in L({\cal A}_1 \downarrow \pi)=L({\cal A}_1 \downarrow \pi_1)$ and  
$w_2 \in L({\cal A}_2 \downarrow \pi)=L({\cal A}_2 \downarrow \pi_2)$ such that 
$w = w_1 \parallel w_2$ in the sense of definition \ref{prod}. Hence, 
 $L[{\cal A}_1 \parallel {\cal A}_2] \downarrow \pi \subseteq 
L({\cal A}_1 \downarrow \pi) \parallel L({\cal A}_2 \downarrow \pi)$. The converse can be shown 
in a similar way.
 \qed
\end{proof}

\subsubsection{Refinement and Parallel Composition}
 \label{sssec_lr_of_cfmsv}
The definition of parallel composition naturally lends itself to a notion of
addition of conformance mappings between design and requirement pairs. 
Consider FSMr's $R_1, R_2$ corresponding to two features $f_1, f_2$. Let 
$D_1, D_2$ be the corresponding FSMd's. Let $\rho^r_1, \rho^r_2$ be the global predicates 
of $R_1, R_2$, and let  $\rho^d_1, \rho^d_2$ be the global predicates 
of $D_1, D_2$ respectively. 
Assume that $D_1 \leq_{\Phi_1}R_1$ and 
$D_2 \leq_{\Phi_2}R_2$. 
Let $\rho^r=\rho^r_{12} \wedge \rho^r_1 \wedge \rho^r_2$ be the 
global predicate of $R_1 \parallel R_2$; likewise, let 
 $\rho^d=\rho^d_{12} \wedge \rho^d_1 \wedge \rho^d_2$ be the 
global predicate of $D_1 \parallel D_2$.  We now want to ask if $D_1 \parallel D_2$ 
conforms to $R_1 \parallel R_2$. This amounts to computing a conformance mapping between 
$D_1 \parallel D_2$ and $R_1 \parallel R_2$ given $\Phi_1, \Phi_2$.
Consider any valid configuration $\pi^d$ of $D_1 \parallel D_2$. By Lemma \ref{lemma:comp}, 
we can write $\pi^d$ as $\pi^d_1+\pi^d_2$, where $\pi^d_1, \pi^d_2$ are valid configurations of $D_1, D_2$ respectively. 
Since $D_1 \leq_{\Phi_1}R_1$ and 
$D_2 \leq_{\Phi_2}R_2$, there exists valid configurations $\pi^r_1 \in \Phi_1(\pi^d_1)$ and 
$\pi^r_2 \in \Phi_2(\pi^d_2)$ in $R_1, R_2$ respectively. Given this, the addition 
of $\Phi_1, \Phi_2$ is defined as follows:

\begin{defn}[Addition of conformance mappings]
\label{def_addconfmapp}
The addition of conformance mappings $\Phi_1, \Phi_2$ is defined to be a mapping 
$\Phi=\Phi_1+\Phi_2$ as follows. 
For every valid configuration $\pi^d=\pi^d_1+\pi^d_2$ of $D_1 \parallel D_2$, 
\begin{eqnarray*}
\Phi(\pi^d)=\{\pi^r \mid \pi^r~\mbox{is a valid configuration of}~
R_1 \parallel R_2, \pi^r=\pi_1^r+\pi_2^r\\
~\mbox{for valid configurations}~\pi^r_1 \in \Phi_1(\pi^d_1),  \pi^r_2 \in \Phi_2(\pi^d_2)\}
\end{eqnarray*}
\end{defn}

\begin{lem}[Conformance of composition]
\label{lemma:conform}
Let $R_1$ and $R_2$ be two FSMr machines corresponding to features $f_1, f_2$, and let 
$D_1$ and $D_2$ be the corresponding FSMd machines.  
Let $D_1 \le_{\Phi_1} R_1$ and $D_2\le_{\Phi_2} R_2$.
Let $\Phi=\Phi_1+\Phi_2$ and $\pi^d$ be a valid configuration of 
$D_1\parallel D_2$. Then,  $\forall \pi^r \in \Phi(\pi^d)$, 
$L([(D_1 \parallel D_2) \downarrow \pi^d]) \subseteq L([ (R_1\parallel  R_2) \downarrow \pi^r])$.
\end{lem}
\begin{proof}
 Given a valid configuration $\pi^d$ of $D_1 \parallel D_2$, we can write it as 
$\pi^d_1+\pi^d_2$, where $\pi^d_1, \pi^d_2$ are respectively valid configurations of $D_1, D_2$ (Lemma \ref{lemma:comp}). 
Since 
$D_1 \leq_{\Phi_1} R_1$ and $D_2 \leq_{\Phi_2} R_2$, there exist valid configurations $\pi^r_1 \in \Phi_1(\pi^d_1)$ and 
$\pi^r_2 \in \Phi_2(\pi^d_2)$ such that $L(D_1 \downarrow \pi^d_1) \subseteq L(R_1 \downarrow \pi^r_1)$ and 
$L(D_2 \downarrow \pi^d_2) \subseteq L(R_2 \downarrow \pi^r_2)$.

Since $\Phi$ has been computed, for every valid configuration $\pi^d$ 
of $D_1 \parallel D_2$, there exists some valid configuration $\pi^r$ of $R_1 \parallel R_2$, $\pi^r \in \Phi(\pi^d)$.
As $\pi^r$ is valid,  $\pi^r \models \rho^r_{12} \wedge \rho^r_1 \wedge \rho^r_2$; 
hence,  $\pi^r$ can be written as  
$\pi^r_1+\pi^r_2$, where $\pi^r_1, \pi^r_2$ are respectively valid configurations of $R_1, R_2$ (Lemma \ref{lemma:comp}), 
and $\pi^r_1 \in \Phi_1(\pi^d_1)$, $\pi^r_2 \in \Phi_2(\pi^d_2)$ by definition \ref{def_addconfmapp}.

$L([(D_1 \parallel D_2) \downarrow \pi^d])=L(D_1 \downarrow \pi^d_1) \parallel L(D_2 \downarrow \pi^d_2)$
by lemma \ref{lemma:comp1}. Similarly,   
$L([(R_1 \parallel R_2)\downarrow \pi^r])=L(R_1 \downarrow \pi^r_1) \parallel L(R_2 \downarrow \pi^r_2)$.  This along with 
the observation that $L(D_1 \downarrow \pi^d_1) \subseteq L(R_1 \downarrow \pi^r_1)$ and 
$L(D_2 \downarrow \pi^d_2) \subseteq L(R_2 \downarrow \pi^r_2)$ gives 
$L([(D_1 \parallel D_2)\downarrow \pi^d]) \subseteq 
L([(R_1 \parallel R_2)\downarrow \pi^r])$.
\qed
\end{proof}


Considering the example, in the FSMr $Req_{dl} \parallel Req_{du}$ 
with  $\rho_r: DL\_Enable=DU\_Enable \wedge Transmission_{dl}=Transmission_{du}$, 
Any configuration where $DL\_Enable=Enable$ but $DU\_Enable=Disable$ is invalid.
However, $\Phi(\tuple{Auto,Speed})$ contains only configurations where $DL\_Enable=Enable$,
$\Phi'(\tuple{Moff,Poff})$ contains only configurations where $DU\_Enable=Disable$
and $\tuple{Auto,Speed}+\tuple{Moff,Poff}$ is a valid configuration of 
$Des_{dl} \parallel Des_{du}$. So the design does not conform to the requirement. 
 However, if we make the extra assumption 
that $\rho_d:Cp1=Moff\wedge Cp2=Poff \Leftrightarrow Cp3=Moff\wedge Cp4=Poff$, 
then $\tuple{Auto,Speed}$ and $\tuple{Moff,Poff}$ are not 
compatible anymore and as a result the design conforms to the requirement.

\subsection{Conformance Checking}
\label{ssec_splmodeling}

Let $F=\{f_1,...,f_n\}$ be a set of features and ${\cal F}$ be the complete system 
comprising the features in $F$, along with the relations between the features. 
Let $R_{i}$ be the FSMr modeling the expected 
behavior and variability of $f_i$, and $D_{i}$ the FSMd extracted from the design of $f_i$.  
Let $\rho_{12 \dots n}^r$ and $\rho_{12 \dots n}^d$ be the compositional 
predicates for $R_1 \parallel \dots \parallel R_n$ and $D_1 \parallel \dots \parallel D_n$ respectively. 
Now we  state the variability conformance problem for an SPL as follows:
Does there exist a conformance mapping $\Phi$ such that $
D_1 \parallel \dots \parallel D_n \le_{\Phi} R_{1} \parallel \dots \parallel \dots R_n$?
A compositional approach to solve the problem is to:\\
(i) check whether the design of every feature conforms to its requirement using Algorithm~1; 
(ii) check whether every valid configuration of $D_1 \parallel \dots \parallel D_n$ can be mapped to a valid configuration of
$R_1 \parallel \dots \parallel R_n$. This is the conformance condition.
\subsection{Checking Conformance Using QBF} 
We implement the second check using QBF solving. Given FSMd's $D_1, \dots, D_n$ and 
FSMr's $R_1, \dots, R_n$, \\
(1) Let $Var(D_i)=\{v^d_{i1}, \dots, v^d_{in}\}$ be the set of variables of design $D_i$, and 
$Var(R_i)=\{v^r_{i1}, \dots, v^r_{im}\}$, the set of variables of requirement $R_i$.
Let $\pi^d: (v^d_{i1}=a_1, \dots, v^d_{in}=a_n)$ be a configuration of $D_i$. 
 We denote by $\pi^d_i(x_{i1}, \dots, x_{in})$ 
 a formula which takes $n$ values 
from $Dom(D_i), 1 \leq i \leq n$ as arguments. 
 If $(v^d_{i1}=a_1, \dots, v^d_{in}=a_n)$ is a chosen assignment, then 
$\pi^d_i(x_{i1}, \dots, x_{in})$  is the conjunction $\bigwedge_{j=1}^n(x_{ij}=a_j)$; \\
 (2) Given $n$ FSMd's and $n$ FSMr's check if $D_i$ conforms to $R_i$ for all $1 \leq i \leq n$ using Algorithm 1.
 This gives the map $\Phi_i$. Assume $\Phi_i(\pi^d_i)=\{\pi^r_{i1}, \dots, \pi^r_{im}\}$, 
where each of $\pi^r_{i1}, \dots, \pi^r_{im}$ are configurations of $R_i$, that have been mapped by $\Phi_i$ 
to some configuration $\pi_i^d$ of $D_i$.  \\
 (3) We encode the above conformance mapping using the formula \\
$\Phi_i(x_{i1},x_{i2}, \dots, x_{in})= \bigvee_{j=1}^m \pi_{ij}^r(y_{i1}, \dots, y_{il})$, where 
$x_{ij}$ takes values from 
$Dom(v^d_{ij})$, and $y_{ij}$ from  $Dom(v^r_{ij})$.\\
(4) Let $\varphi^d_{i,j}=\rho^d \wedge \rho^d_i \wedge \rho^d_j$ and 
$\varphi^r_{i,j}=\rho^r \wedge \rho^r_i \wedge \rho^r_j$ represent respectively the propositional formulae  
which ensures consistency of the global predicates of $D_i, D_j$ 
and $R_i, R_j$ along with the compositional predicates $\rho^d$ and $\rho^r$.
Given a set $S \subseteq \{1, 2, \dots, n\}$,  $\varphi^d_S$ and 
$\varphi^r_S$ can be appropriately written.\\
The QBF formula 
for conformance checking is given by 
$$\Psi=\forall x_{11} \dots x_{ni_n}[
\varphi^d_{1,2,\dots,n}
 \Rightarrow 
\exists y_{11} \dots y_{nj_n}( 
 \Phi_1 \wedge \dots \wedge \Phi_n \wedge \varphi^r_{1,2, \dots, n})]$$

\begin{theorem}
\label{thm:main}
 Given a SPL, let $\{f_1, \dots, f_n\}$ be the set 
of features in a chosen product. Let  $D_i, R_i$ be the FSMd and FSMr for feature $f_i$. 
Then $D_1 \parallel \dots \parallel D_n$ conforms to $R_1 \parallel \dots \parallel R_n$ iff 
$\Psi$ holds. 
\end{theorem}
\begin{proof}
Given $D_i \leq_{\Phi} R_i$, assume that $D_1 \parallel \dots \parallel D_n$ conforms to $R_1 \parallel \dots \parallel R_n$. 
Then, by definition of conformance, it means that for all valid configurations $\pi^d$ of 
$D_1 \parallel \dots \parallel D_n$, there exists a valid configuration $\pi^r$ of $R_1 \parallel \dots \parallel R_n$ such that
$L([D_1 \parallel \dots \parallel D_n]\downarrow \pi^d) \subseteq L([R_1 \parallel \dots \parallel R_n]\downarrow \pi^r)$. 
Let $\Phi$ be the conformance mapping such that $\pi^r$ $\in \Phi(\pi^d)$. 
 
$\pi^d$ is a valid configuration of $D_1 \parallel \dots \parallel D_n$ implies 
that $\pi^d \models \bigwedge_{S \subseteq \{1,2,\dots,n\}} \rho^d_S$, where $\rho^d_S$ is the 
global predicate of $D_{i_1} \parallel \dots \parallel D_{i_j}$, when $S=\{i_1, \dots, i_j\}$. 
Using Lemma \ref{lemma:comp} repeatedly, we can then say that 
$\pi^d=\pi^d_1 + \dots + \pi^d_n$ for valid configurations $\pi^d_i$ of $D_i$. 
Since $\pi^r \in \Phi(\pi^d)$, by definition of conformance mappings, $\pi^r$ must be a valid
 configuration of $R_1 \parallel \dots \parallel R_n$, hence $\pi^r=\pi^r_1+\dots+\pi^r_n$ (Lemma \ref{lemma:comp}),
such that 
$\pi^d_i \in \Phi(\pi^r_i)$, for valid configurations $\pi^r_i$ of $R_i$. 
$\pi^r$ is valid means $\pi^r \models \bigwedge_{S \subseteq \{1,2,\dots,n\}} \rho^r_S$.

Given the above, we show that the QBF $\Psi$ holds. 
The LHS of the QBF $\Psi$ is the formula $\varphi_{1,2,\dots,n}^d$, which  is the conjunction 
$\rho^d_S$ for all subsets $S$ of $\{1,2,\dots,n\}$. The forall quantifier outside would 
thus evaluate all configurations of $D_1 \parallel \dots \parallel D_n$ that satisfy 
$\varphi_{1,2,\dots,n}^d$; that is, which satisfy 
$\bigwedge_{S \subseteq \{1,2,\dots,n\}} \rho^d_S$ : hence, all valid configurations of $D_1 \parallel \dots \parallel D_n$.

For the QBF to hold good, for all valid configurations of $D_1 \parallel \dots \parallel D_n$ that have been 
evaluated on the LHS, we must find some configuration of $R_1 \parallel \dots \parallel R_n$ that 
satisfies $\Phi_1 \wedge \dots \wedge \Phi_n \wedge \varphi_{1,2,\dots,n}^r$ : 
(i) any configuration $\pi$ of $R_1 \parallel \dots \parallel R_n$ that satisfies 
$\varphi_{1,2,\dots,n}^r$ would be valid; (ii) further, 
if it has to satisfy   $\Phi_1 \wedge \dots \wedge \Phi_n$, 
it must agree with $\pi^r_i \in \Phi_i(\pi^d_i)$
over $Var(R_i)$ for all $1 \leq i \leq n$. By Lemma \ref{lemma:comp}, this means that 
$\pi$ can be written as $\pi^r_1+\dots+\pi^r_n$. Thus, for the QBF to hold, 
we must be able to find for each valid configuration 
$\pi^d$ of $D_1 \parallel \dots \parallel D_n$, 
a valid configuration $\pi^r$
of $R_1\parallel \dots \parallel R_n$ which can be written as 
$\pi^r_1+\dots+\pi^r_n$, where $\pi^r_i \in \Phi_i(\pi^d_i)$ for each $i$. 
But this is exactly what the mapping $\Phi$ which 
checks for conformance of $D_1 \parallel \dots \parallel D_n$ with 
$R_1 \parallel \dots \parallel R_n$ does. Since we assume that $\Phi$ exists, 
the QBF holds.

The converse can be shown in a similar way : that is, if the QBF formula $\Psi$ holds, then 
$D_1 \parallel \dots \parallel D_n$ will conform to $R_1 \parallel \dots \parallel R_n$.
\qed
\end{proof}

\section{Implementation and Case Studies}
\label{ssec_ecpl}

\begin{figure}[hb]
 \begin{center}
   {\includegraphics[width=0.7\textwidth]{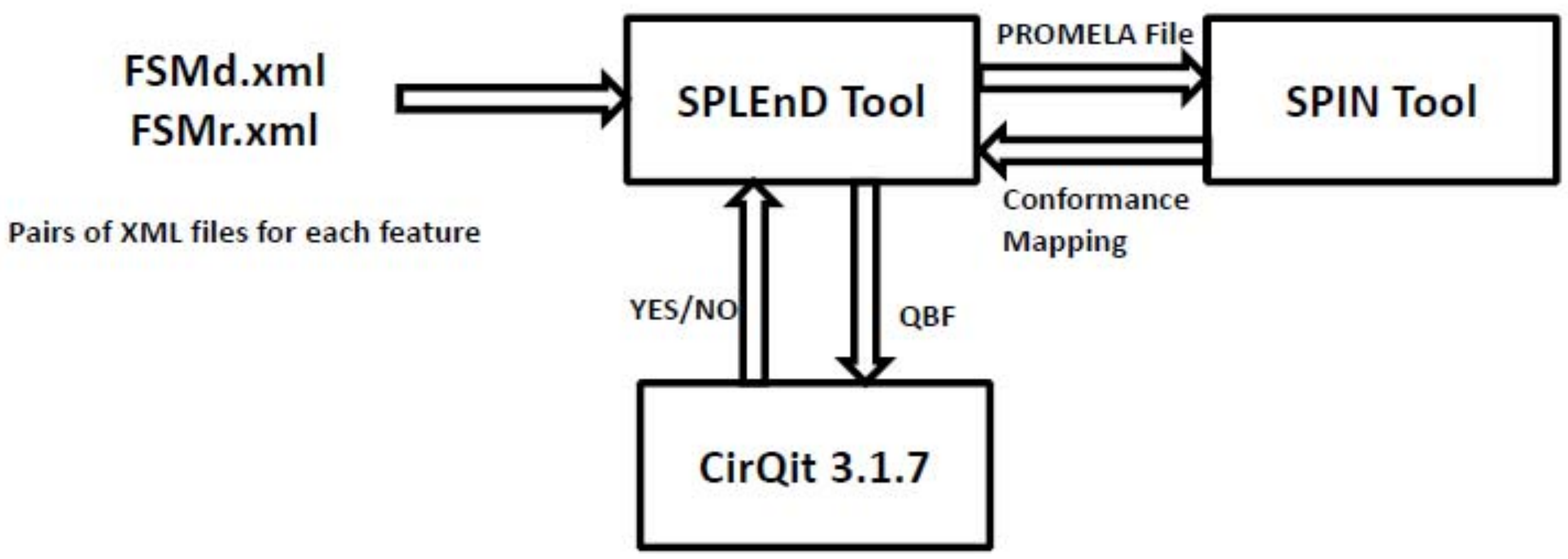}}
\vspace{-2.5cm}
  \caption{Overview of SPLEnD} 
  \label{splend}
  \end{center}
\end{figure}

Figure \ref{splend} pictorially describes the tool SPLEnD. It takes as input, a 
pair of xml files corresponding to FSMd, FSMr and outputs 
a PROMELA file. The latter is fed to SPIN, which returns the conformance 
mappings, or declares non-conformance; given the conformance mapping the 
tool computes a QBF formula $\Psi$ which is fed to CirQit. 

We considered two real case studies for our experimentation: 
Entry Control Product Line,  ECPL having 7 features and  
Banking Software Product Line, BSPL, composed of  25 features. 
The details of the ECPL and BSPL case 
studies are given below.  The FSMr, FSMd models of each feature contains less than $10$ states.

\section{ECPL and BSPL}
In this section, we describe the two product lines that have been considered in the paper : (i) ECPL and 
(ii) BSPL.
\subsection{ECPL}
\label{ecpl} 
The Entry Control Product Line comprises all the features involved in the 
management of the locks in a car. In this study, we focus on the following 
features:
\begin{itemize}
 \item {\em Power lock}: this is the basic locking functionality which manages 
the locking/unlocking according to key button press and courtesy switch press,  
\item {\em Last Door Closed Lock}: delays the locking of the doors until all 
the doors are closed. It is applicable when the lock command appends while a door is open,
\item {\em Door lock}: automates the locking of doors when the vehicle starts, 
\item {\em Door unlock}: automates the unlocking of door(s) when the vehicle 
stops,
\item {\em Anti-lockout}: is intended to prevent the inadvertent lockout 
situations: the driver is out of the car with the key inside and all the doors 
locked, 
\item {\em Post crash unlock}: unlocks all the doors in a post crash situation,
\item {\em Theft security lock}: secures the car with a second lock.\\
\end{itemize}
Each feature is represented as a pair of state machines containing $3$ to $10$ states.
\subsubsection{The variability constraints of the ECPL}
Figure \ref{figure_fm} 
presents the feature diagram of the ECPL (a la Czarnecki \cite{Czarnecki00}).
 This diagram presents the variability constraints of the ECPL at the requirement level ($\rho_{f_0}$).
 All the constraints represented  by this diagram have to be considered during composition to guarantee the 
overall consistency of the SPL behavior. 
The dark gray boxes are features of the ECPL: 
{\em Power lock}, {\em Anti-lockout}, {\em Door lock}, {\em Door unlock}, and {\em Post crash unlock}. 
The light gray boxes are configurations. The black arrow from the ``Manual" configuration to the ``Shift out of park" 
configuration and to the ``Shift into park" configuration says that if the transmission is manual, 
the targeted configurations cannot be selected. i.e. In ``Manual" configuration, there is no ``park" gear.

\begin{figure}[htbp]
 \begin{center}
   {\includegraphics[width=0.9\textwidth]{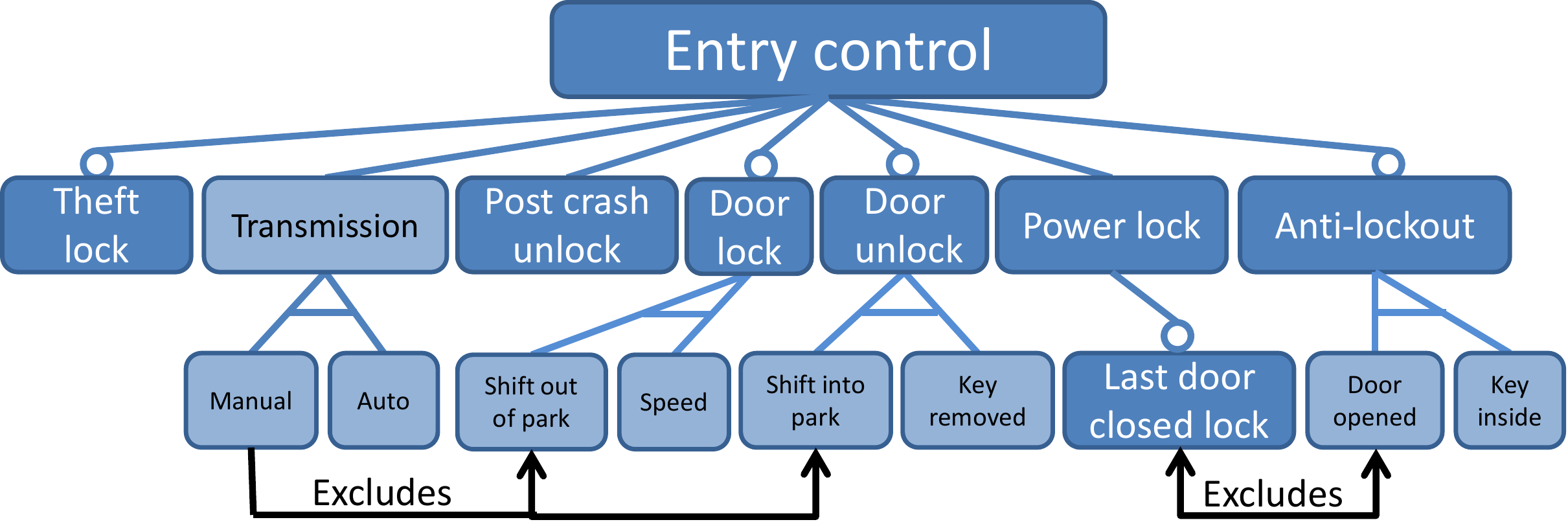}}
  \caption{The feature diagram of the ECPL.} 
  \label{figure_fm}
  \end{center}
\end{figure}

\subsection{BSPL}
\label{bspl}

The Banking Software Product Line (BSPL) consists of 25 behavioral features. 
The BSPL is used to derived the software for ATM, Bank, Online Banking and Mobile Banking.
 Figure \ref{figure_bspl_fm} presents the feature diagram of the BSPL.

\begin{figure}[!ht]
 \begin{center}
   {\includegraphics[scale=0.4]{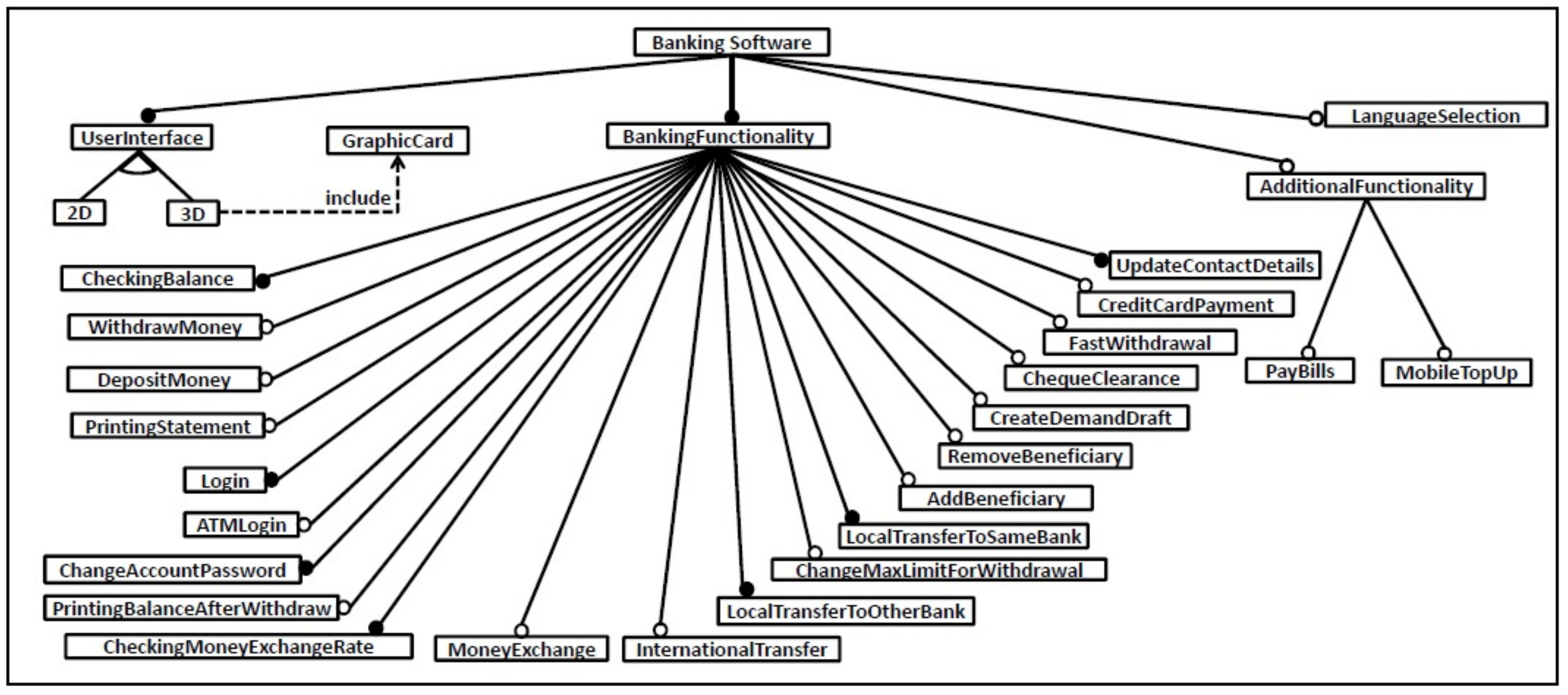}}
  \caption{The feature diagram of the BSPL.} 
  \label{figure_bspl_fm}
  \end{center}
\end{figure}

\vspace{0.5in}
Similar to ECPL, we ran Algorithm 1 on all the 25 features of BSPL. 
In section \ref{ssec_ecpl}, Figure \ref{bsplResultsTable} presents the number of design 
configurations and execution time of Algorithm 1 for each feature. 
In the following, we elaborate on the FSMv of 2 features: (i) User Interface and (ii) Withdraw Money. 
The FSMd/FSMr for all the features has states between 2 and 10 (both inclusive).
Figure \ref{figure_ui_fsmr} is the FSMr for feature $User$ $Interface$, 
which has $UI$ as an event with global predicate 
$\rho = \{ \neg (uip = Disable)\}$. There is only one boolean variable, $Var = \{uip\}$, $uip$ takes values from $\{Enable, Disable\}$.
\begin{figure}[htbp]
 \begin{center}
   {\includegraphics[scale=0.15]{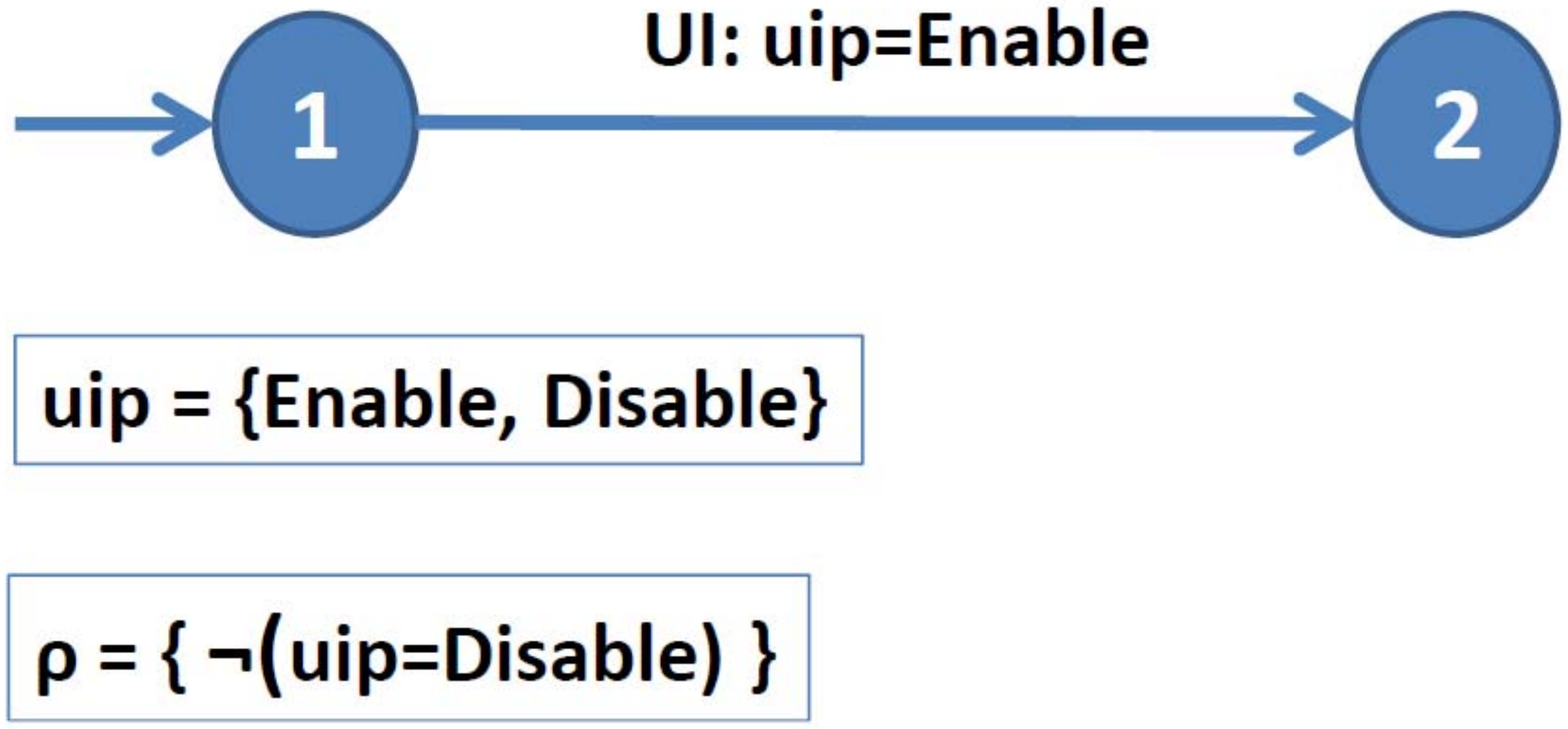}}
  \caption{FSMr for feature: $UserInterface$.} 
  \label{figure_ui_fsmr}
  \end{center}
\end{figure}

Figure \ref{figure_ui_fsmd} is the FSMd for feature $User$ $Interface$. 
This FSMd shares the event $UI$ with the FSMr and
has global predicate
 $\rho = \{ (type = 2D \vee type = 3D)\}$. 
There are two variables, $Var = \{type, graphics\}$, $type$ takes values from $\{2D, 3D\}$, while $graphics$ takes values from 
$\{Enable, Disable\}$.
\begin{figure}[htbp]
 \begin{center}
   {\includegraphics[scale=0.2]{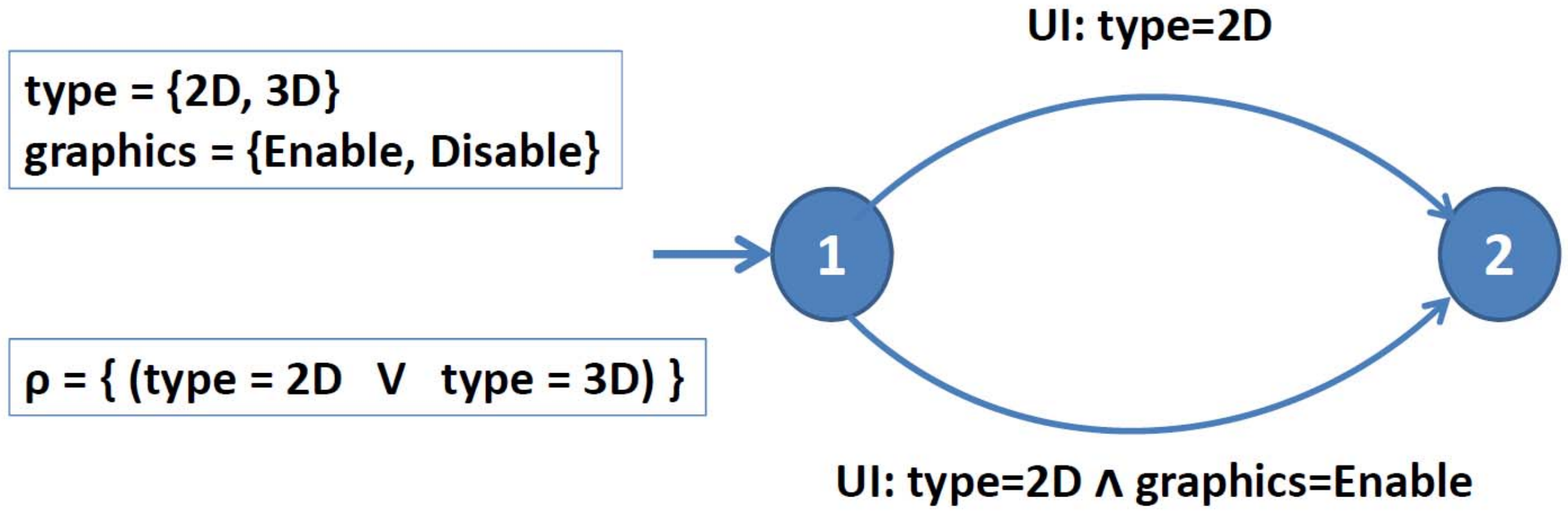}}
  \caption{FSMd for feature: $UserInterface$.} 
  \label{figure_ui_fsmd}
  \end{center}
\end{figure}

The analysis results for the two case studies are summarized in 
Figures \ref{expres} and \ref{bsplResultsTable} which gives the  
times taken by Algorithm 1. The number of product variants and the time taken
for Algorithm 1 are very small in both case studies. In the case of ECPL,  
a bug was found in the feature $Door~Lock$ \footnote{In $Des_{dl}$, the 
transition from the middle elliptical state to the round state labeled with 
$Poff: ShiftOutOfPark$ is incorrect; $\Phi(\tuple{Auto,Poff}) = \emptyset$. 
Removing this transition fixes the bug.}.
In this case, after fixing the bug, for the second step we used SPIN which took
11 seconds. 
For BSPL, the second step was performed using the QBF approach and CirQit 
took just 0.005 seconds. 
\begin{figure}[htbp]
	\centering
		\begin{tabular}{|c|c|c|c|}
			\hline {Sr. No.} & {Features}   &  Design Variants & 
SPIN Time(Sec)\\
			\hline 1 & UserInterface & 6 & 0.002 \\
			\hline 2 & CheckingBalance & 3 & 0.003 \\
			\hline 3 & WithdrawMoney & 8 & 0.027 \\
			\hline 4 & DepositMoney & 2 & 0.002 \\
			\hline 5 & PrintingStatement & 3 & 0.002  \\
			\hline 6 & Login & 1 & 0.001 \\
			\hline 7 & ATMLogin & 1 & 0.001  \\
			\hline 8 & ChangeAccountPassword & 2 & 0.003 \\
			\hline 9 & PayBills & 2  & 0.003  \\
			\hline 10 & PrintingBalanceAfterWithdraw & 2  & 0.003  \\
			\hline 11 & CheckingMoneyExchangeRate & 2 & 0.003  \\
			\hline 12 & MoneyExchange & 2  & 0.004  \\
			\hline 13 & InternationalTransfer & 2  & 0.006  \\
			\hline 14 & LocalTransferToOtherBank & 1  & 0.004  \\
			\hline 15 & LanguageSelection  & 2 & 0.001 \\
			\hline 16 & MobileTopUp  & 2 & 0.002  \\
			\hline 17 & ChangeMaxLimitForWithdrawal & 1  &  0.003 \\
			\hline 18 & LocalTransferToSameBank & 3 & 0.003  \\
			\hline 19 & AddBeneficiary  & 1  & 0.002  \\
			\hline 20 & RemoveBeneficiary  & 1  & 0.002  \\
			\hline 21 & CreateDemandDraft & 2  & 0.003 \\
			\hline 22 & ChequeClearance  & 1  & 0.003  \\
			\hline 23 & FastWithdrawal  & 1  & 0.002  \\
			\hline 24 & CreditCardPayment  & 2  & 0.002  \\
			\hline 25 & UpdateContactDetails  & 2  & 0.004  \\
			\hline
	\end{tabular}
\caption{Execution~time~of~FSMv-Verifier~on~Algorithm~1 for BSPL}
	\label{bsplResultsTable}
\end{figure}

\begin{figure}[htbp]
	\centering
		\begin{tabular}{|c|c|c|c|c|c|c|}
			\hline
			{Features}  & {\em PL} \& {\em LDCL}& {\em PCU}& {\em  DL}& {\em DU} & {\em AL} & {\em TSL}\\
			\hline
			{Design Variants}& 8 & 3 & 4 & 7 & 3 & 8 \\
			\hline
			{SPIN Time (Sec)}  & 0.436 & 0.031 & 0.046 & 0.109 & 0.015 & 0.218 \\
			\hline
	\end{tabular}
\caption{Execution~time~of~FSMv-Verifier~on~Algorithm~1 for ECPL}
	\label{expres}
\end{figure}

In the automotive domain, really very large SPLs are constructed~\cite{splc12}. Before undertaking the task of modeling such large examples, in order to 
quickly determine the scalability of our approach, we generated many random 
SPLs with 5000 to 25,000 features. Each of the corresponding FSMr/FSMd has two 
variables (four variants), and $3$ to  $8$ states. 
Similar to the ECPL and BSPL cases, SPIN took very little time (less than 0.5 seconds) for each (FSMr, FSMd) pair.  
The composite FSMr/FSMd, and hence the QBF formula $\Psi$ has then 10,000 to 
50,000 variables.  As we can see from Figure \ref{qbfScaleResults}, the 
the time taken for the largest example is 196.69 seconds which is quite 
efficient. Encouraged by this result, we plan to take up the large industrial 
case studies.

\begin{figure}[htbp]
	\centering
		\begin{tabular}{|c|c|c|c|c|c|}
		\hline 
Variables in FSMr/FSMd & 10000 & 20000 & 30000 & 40000 & 50000 \\
\hline
CirQit 3.1.7 time (Sec) & 4.47 & 25.77 & 65.67 & 119.49 & 196.69\\ 
			\hline
	\end{tabular}
\caption{Execution time of QBF for Scalability}
	\label{qbfScaleResults}
\end{figure}

\section{Conclusion}
\label{sec_conclu}
This paper motivated the need for extending the classical design verification
problem to evolving SPL in which features and variability information can be 
added incrementally. The novel aspects of the proposed work are: 
(i) it verifies that the variability at the design level conforms to that at 
the requirement level, (ii) it is compositional and (iii) it reduces
the conformance checking problem to QBF sat solving. 
A prototype tool has been implemented and experimented with modest sized 
examples with encouraging results. 


\bibliography{mainFSMv}

\end{document}